\documentclass[10pt, amsart]{article}

\usepackage{amsmath}
\usepackage{amsthm} 
   \usepackage{amssymb,amsfonts,textcomp}
   \usepackage{bm}

\usepackage[T3,T1]{fontenc} 
\usepackage{txfonts} 

\usepackage{latexsym} 
\usepackage{mathrsfs} 

\usepackage{array}
\usepackage{hhline}
\usepackage{graphicx}

\usepackage{tikz}
\usetikzlibrary{trees}
\usepackage{qtree}

\usetikzlibrary{decorations.pathmorphing}
\usetikzlibrary{decorations.markings}
\usepackage{verbatim}

\usepackage{tikz}
\usetikzlibrary{arrows,shapes,chains,matrix,positioning,scopes}
\usetikzlibrary{decorations.pathmorphing}
\usetikzlibrary{decorations.markings}
\usepgflibrary{plotmarks}
\usetikzlibrary{patterns}

\usepackage{pdfpages}

\usetikzlibrary{mindmap}

\usepackage{graphicx}

\usepackage{multirow}
\usepackage{multicol}

\usepackage{subfig} 

\usepackage{float}
\usepackage{wrapfig}
\restylefloat{figure}
\restylefloat{table}
\usepackage{color, colortbl} 

\usepackage{float}
\restylefloat{figure}

\newtheoremstyle{slplain}
 {.5\baselineskip\@plus.2\baselineskip\@minus.2\baselineskip}
  {.5\baselineskip\@plus.2\baselineskip\@minus.2\baselineskip}
  {\slshape}
  {}
  {\bfseries}
  {.}
  { }
  {}

\theoremstyle{definition}
\theoremstyle{slplain}

\theoremstyle{slplain}

\newtheorem{thm}{Theorem} 
\newtheorem{cor}{Corollary}

\newcommand{\cat}{^\vee \!}

\newcommand{\pitcher}{^\wedge \!}

\newcommand{\catk}{\cat\,^{k}}
\newcommand{\catm}{\cat\,^{m}}
\newcommand{\catZero}{\cat\,^{0}}

\newcommand{\pitcherk}{\pitcher\,^{k}}

\newcommand{\pitcherZero}{\pitcher\,^{0}}

\newcommand{\match}{\mbox{\bf Match}}
\newcommand{\cisM}{\mbox{\bf cis}^{*}} 
\newcommand{\transM}{\mbox{\bf trans}^{*}} 
\newcommand{\metaM}{\mbox{\bf meta}^{*}} 

\newcommand{\fiadrA}{\cat\alpha} 

\newcommand{\foadrB}{\pitcher\beta}

\newcommand{\adrA}{\bf a}
\newcommand{\adrB}{\bf b}
\newcommand{\adrC}{\bf c}
\newcommand{\adrD}{\bf d}
\newcommand{\adrE}{\bf e}

\newcommand{\adrV}{\bf v}
\newcommand{\adrX}{\bf x}
\newcommand{\adrY}{\bf y}

\newcommand{\nodeA}{\mathcal A}
\newcommand{\nodeB}{\mathcal{B}}

\newcommand{\nodeD}{\mathcal{D}}
\newcommand{\nodeV}{\mathcal{V}}

\newcommand{\tfg}{{^\wedge \!}\mbox{{\sf tf}}} 
\newcommand{\tfc}{{^\vee \!}\mbox{{\sf tf}}}
\newcommand{\tfp}{{^\circ \!}\mbox{{\sf tf}}}

\definecolor{currentcolor}{rgb}{0.8 0.4 0.2}

\tikzstyle{stochasticjumpstyle}=[diamond,draw,fill=white,>=latex,>->,dashed]
\tikzstyle{stochasticPathstyle}=[>=latex,>->,dashed]
\tikzstyle{stochasticNodestyle}=[ellipse,inner sep=1pt,text=.,fill=.!20]
\tikzstyle{blankstyle}=[fill=white,inner sep=1pt]

\def\SnakeSegLen{0.6em}
\def\SnakeAmp{0.11em}
\def\PrePostLen{5mm}

\tikzstyle{sendstyle}=[dashed,line width=1.1pt]

\tikzstyle{splitstyle}=[circle,draw]

\tikzstyle{receivestyle}=[>->,line width=1.1pt,decorate, decoration={zigzag,segment length=\SnakeSegLen, amplitude=\SnakeAmp, pre length=\PrePostLen, post=curveto, post length=\PrePostLen},text=black]

\tikzstyle{receivesigstyle}=[draw,inner sep=2pt,fill=pink!20]
\tikzstyle{receivesigstyle3}=[draw,inner sep=2pt, fill=white]
\tikzstyle{receivesigstyle2}=[ellipse,shade, draw,double,fill=red!10]

\tikzstyle{sendsigstyle}=[diamond,draw,inner sep=1pt, text=black, fill=yellow!80]
\tikzstyle{sendsigstyle3}=[circle,draw, ball color=white]
\tikzstyle{sendsigstyle2}=[diamond,draw,double, inner sep=1pt, fill=white]

\tikzstyle{snakesendstyle}=[*->, decorate, decoration={snake, segment length=\SnakeSegLen, amplitude=\SnakeAmp,  pre length=\PrePostLen, post=curveto, post length=\PrePostLen}]

\tikzstyle{snakesendstyle1}=[line width=1.1pt, decorate, decoration={snake,segment length=\SnakeSegLen, amplitude=\SnakeAmp}]

\tikzstyle{snakesendstyle3}=[decorate, decoration={markings, mark=at position .75 with {\arrow[red,line width=5mm]{>}}, snake, segment length=\SnakeSegLen, amplitude=\SnakeAmp,  pre length=\PrePostLen, post=curveto, post length=\PrePostLen}]

\tikzstyle{snakesendstyle2}=[decorate, decoration={ zigzag,segment length=\SnakeSegLen, amplitude=\SnakeAmp, line around/.style={decoration={pre length=\PrePostLen,post length=\PrePostLen}}}]

\newcounter{foo}
\usepackage[bookmarks=false]{hyperref} 

\makeatletter
\AtBeginDocument{\let\nameref\HyPsd@nameref}

\makeatother

\colorlet{anglecolor}{green!50!black}

\definecolor{darkgreen}{rgb}{0 0.6  0}

\definecolor{turquoise}{rgb}{0 0.41 0.41}
\definecolor{rouge}{rgb}{0.79 0.0 0.1}
\definecolor{vert}{rgb}{0.15 0.4 0.1}
\definecolor{mauve}{rgb}{0.6 0.4 0.8}
\definecolor{violet}{rgb}{0.58 0. 0.41}
\definecolor{orange}{rgb}{0.8 0.4 0.2}
\definecolor{bleu}{rgb}{0.39, 0.58, 0.93}

\definecolor{darkross}{rgb}{0.008,0.412,0.471}
\definecolor{middleross}{rgb}{0.012,0.580,0.663}
\definecolor{lightross}{rgb}{0.016,0.749,0.855}
\definecolor{darkblue}{rgb}{0.067,0.008,0.471}
\definecolor{middleblue}{rgb}{0.094,0.012,0.663}
\definecolor{lightblue}{rgb}{0.122,0.016,0.855}
\definecolor{darkpurple}{rgb}{0.471,0.008,0.412}
\definecolor{middlepurple}{rgb}{0.663,0.012,0.580}
\definecolor{lightpurple}{rgb}{0.855,0.016,0.749}
\definecolor{darkbrown}{rgb}{0.471,0.067,0.008}
\definecolor{middlebrown}{rgb}{0.663,0.094,0.012}
\definecolor{lightbrown}{rgb}{0.855,0.122,0.016}
\definecolor{darkolive}{rgb}{0.412,0.471,0.008}
\definecolor{middleolive}{rgb}{0.580,0.663,0.012}
\definecolor{lightolive}{rgb}{0.749,0.855,0.016}
\definecolor{darkgreen}{rgb}{0.008,0.417,0.067}
\definecolor{middlegreen}{rgb}{0.012,0.663,0.094}
\definecolor{lightgreen}{rgb}{0.016,0.855,0.122}
\definecolor{darkocre}{rgb}{0.471,0.298,0.008}
\definecolor{middleocre}{rgb}{0.663,0.420,0.012}
\definecolor{lightocre}{rgb}{0.855,0.541,0.016}

    \definecolor{lightblue}{rgb}{0,0,.7}
    \definecolor{orange}{rgb}{1,.7,0}
    \definecolor{darkorange}{rgb}{1,.4,0}
    \definecolor{darkgreen}{rgb}{0,.5,0}
    \definecolor{darkblue}{rgb}{0,0,.4}
    \definecolor{darkred}{rgb}{.4,0,0}
    \definecolor{gray}{rgb}{.2,.2,.2}
    \definecolor{darkgray}{rgb}{.2,.2,.2}
    \definecolor{shadecolor}{gray}{0.925}
    
\definecolor{darkred}{rgb}{0.65,0,0}
\definecolor{darkblue}{rgb}{0,0,.65}
\definecolor{darkgreen}{rgb}{0,0.5,0}
\definecolor{orange}{rgb}{1,.75,.25}
\definecolor{aqua}{rgb}{0,.25,.75}
\definecolor{grey}{rgb}{.5,.5,.5}
\definecolor{brown}{rgb}{.51,.35,.18}

\definecolor{lightblue}{rgb}{.3,.5,1}
\definecolor{orange}{rgb}{1,.7,0}
\definecolor{darkorange}{rgb}{1,.4,0}
\definecolor{darkgreen}{rgb}{0,.4,0}
\definecolor{darkblue}{rgb}{0,0,.4}
\definecolor{darkred}{rgb}{.56,0,0}
\definecolor{gray}{rgb}{.3,.3,.3}
\definecolor{darkgray}{rgb}{.2,.2,.2}

\definecolor{blue}{rgb}{0,0,1}
\definecolor{red}{rgb}{1,0,0}
\definecolor{pink}{rgb}{.933,0,.933}
\definecolor{green}{rgb}{0.133,0.545,0.133}
\definecolor{shadecolor}{gray}{0.925}

\definecolor{DarkBlue}{rgb}{0.000,0.000,0.545}
\definecolor{DarkChocolate}{rgb}{0.400,0.200,0.000}
\definecolor{DarkCyan}{rgb}{0.000,0.545,0.545}
\definecolor{DarkGoldenrod}{rgb}{0.720,0.525,0.044}
\definecolor{DarkGray}{rgb}{0.664,0.664,0.664}
\definecolor{DarkGreen}{rgb}{0.000,0.392,0.000}
\definecolor{DarkGrey}{rgb}{0.664,0.664,0.664}
\definecolor{DarkKhaki}{rgb}{0.740,0.716,0.420}
\definecolor{DarkLavender}{rgb}{0.400,0.200,0.600}
\definecolor{DarkMagenta}{rgb}{0.545,0.000,0.545}
\definecolor{DarkOliveGreen}{rgb}{0.332,0.420,0.185}
\definecolor{DarkOrange}{rgb}{1.000,0.550,0.000}
\definecolor{DarkOrchid}{rgb}{0.600,0.196,0.800}
\definecolor{DarkPeriwinkle}{rgb}{0.400,0.400,1.000}
\definecolor{DarkPurpleBlue}{rgb}{0.400,0.000,0.800}
\definecolor{DarkRed}{rgb}{0.545,0.000,0.000}
\definecolor{DarkRoyalBlue}{rgb}{0.000,0.200,0.800}
\definecolor{DarkSalmon}{rgb}{0.912,0.590,0.480}
\definecolor{DarkSeaGreen}{rgb}{0.560,0.736,0.560}
\definecolor{DarkSlateBlue}{rgb}{0.284,0.240,0.545}
\definecolor{DarkSlateGray}{rgb}{0.185,0.310,0.310}
\definecolor{DarkSlateGrey}{rgb}{0.185,0.310,0.310}
\definecolor{DarkSmoke}{rgb}{0.920,0.920,0.920}
\definecolor{DarkTurquoise}{rgb}{0.000,0.808,0.820}
\definecolor{DarkViolet}{rgb}{0.580,0.000,0.828}
\definecolor{DeepPink}{rgb}{1.000,0.080,0.576}
\definecolor{DeepSkyBlue}{rgb}{0.000,0.750,1.000}

\tikzstyle{mystyle}=[scale= \PicSize,  
baseline=(current bounding box.center),
nodedescr/.style={fill=white,inner sep=2.5pt},
nodedescr2/.style={draw,circle,fill=white},
description/.style={fill=white,inner sep=2pt},
cancer loop/.style={preaction={draw=white, -,line width=6pt}, line width=1.1pt},
pot1 loop/.style={preaction={draw=white, -,line width=6pt}, red, line width=1.1pt},
pot2 loop/.style={preaction={draw=white, -,line width=6pt}, green, line width=1.1pt},
pot1/.style={preaction={draw=white, -,line width=6pt}, gray, solid, line width=1pt}, 
pot2/.style={preaction={draw=white, -,line width=6pt}, gray, solid, line width=1pt},
inPot/.style={ out=45,in=90,distance=3cm, min distance=2cm, line width=1.1pt, black!75},
inPot1/.style={ out=45,in=120,distance=3cm, min distance=2cm, line width=1.1pt, black!75},
inPot2/.style={ out=-45,in=-120,distance=3cm, min distance=2cm, line width=1.1pt, black!75},
in100Pot1/.style={ out=80,in=90,distance=3cm, min distance=2cm, line width=1.1pt, darkgreen},
in100Pot2/.style={ out=-80,in=-90,distance=3cm, min distance=2cm, line width=1.1pt, brown},
in2Pot1/.style={ out=45,in=90,distance=3cm, min distance=2cm, line width=1.1pt, black!75},
in2Pot2/.style={ out=-65,in=-90,distance=3cm, min distance=2cm, line width=1.1pt, black!75},
selfloop1/.style={ out=45,in=120,distance=6cm, min distance=4cm, line width=1.1pt, red},
selfloop2/.style={ out=-45,in=-120,distance=6cm, min distance=4cm, line width=1.1pt, blue},
loop1red/.style={ out=45,in=120,distance=6cm, min distance=4cm, line width=1.1pt, red},
loop2red/.style={ out=-45,in=-80,distance=6cm, min distance=4cm, line width=1.1pt, red},
cross line/.style={preaction={draw=white, -,	line width=6pt}},
jump1/.style={ out=45,in=135,distance=4cm, min distance=2cm, line width=1.1pt, red, dashed},
jump2/.style={ out=-45,in=-135,distance=4cm, min distance=2cm, line width=1.1pt, blue, dashed},
jumpover/.style={ out=80,in=90,distance=4cm, min distance=3cm, line width=1.1pt, black, dashed},
stochastic jump/.style={>=latex,>->,dashed, line width=1.1pt, green}
]

\def\PicSize{0.5} 

\def\nexttoPicSize2{6.0cm}

\numberwithin{equation}{section}

\usepackage[margin=1.4in]{geometry} 

\usepackage[parfill]{parskip}    

\begin{document}

\title{\vspace{-7ex}Combinatorial Limits of Transcription Factors and Gene Regulatory Networks in Development and Evolution}

\author{Eric Werner \thanks{Balliol Graduate Centre, Oxford Advanced Research Foundation (http://oarf.org).  
\copyright Eric Werner 2015.  All rights reserved. }\\ \\
University of Oxford\\
Department of Physiology, Anatomy and Genetics, \\
and Department of Computer Science, \\
Le Gros Clark Building, 
South Parks Road, 
Oxford OX1 3QX  \\
email:  eric.werner@dpag.ox.ac.uk\\
}

\date{ } 
\maketitle
\thispagestyle{empty}

\begin{center}
\textbf{Abstract}
\begin{quote}
\it \small

Gene Regulatory Networks (GRNs) consisting of combinations of transcription factors (TFs)  and their cis promoters are assumed to be sufficient to direct the development of organisms.  Mutations in GRNs are assumed to be the primary drivers for the evolution of multicellular life.  Here it is proven that neither of these assumptions is correct.  They are inconsistent with fundamental principles of combinatorics of bounded encoded networks.   It is shown there are inherent complexity and control capacity limits for any gene regulatory network that is based solely on protein coding genes such as transcription factors.  This result has significant practical consequences for understanding development, evolution, the Cambrian Explosion, as well as multi-cellular diseases such as cancer.   If the arguments are sound, then genes cannot explain the development of complex multicellular organisms and genes cannot explain the evolution of complex multicellular life. 
\end{quote}
\end{center}


\pagebreak
\maketitle
\tableofcontents

\pagebreak

\section{Introduction}
Gene Regulatory Networks (GRNs) consisting of combinations of transcription factors (TFs)  and their cis promoters are assumed to be sufficient to direct the development of organisms.  Mutations in GRNs are assumed to be the primary drivers for the evolution of multicellular life.  Here it is proven that neither of these assumptions is correct.  They are inconsistent with fundamental principles of combinatorics of bounded encoded networks.   It is shown there are inherent complexity and control capacity limits for any gene regulatory network that is based solely on protein coding genes such as transcription factors.  This result has significant practical consequences for understanding development, evolution, the Cambrian Explosion, as well as multi-cellular diseases such as cancer.   If the arguments are sound, then genes cannot explain the development of complex multicellular organisms and genes cannot explain the evolution of complex multicellular life\footnote{This paper gives a more formal proof of the informal proof given in (Werner, E., "What Transcription Factors Can't Do: On the Combinatorial Limits of Gene Regulatory Networks" arXiv:1312.5565 [q-bio.MN], 2013.) However, the concepts and arguments are just as valid in the informal proof as in this more formal version. Even though the paper is still rough and somewhat incomplete, I put  this out there for feedback from the life science, mathematics, and, more generally, the science communities. }

\section{Addressing  networks}

An   {\em addressing network} $N$ is an address-based network that consists of a set of nodes with addresses. The addresses define the network's edges or links when addresses of two nodes match.  Formally, an addressing network is a tuple $N = (\mathbb{D, I, O, \match, X})$ where $\mathbb{D}$ is a set of {\em nodes}. $\mathbb{I}$ is a set of {\em unitary In-addresses}. $\mathbb{O}$ is a set of {\em unitary Out-addresses}.   $\match\subseteq (\mathbb{O} \times \mathbb{I})$ is a {\em matching relation} between unitary Out-addresses and unitary In-addresses.  $\mathbb{X}$ is a set of {\em actions}.  Unitary addresses are considered primitive, indivisible units that combine to form address combinations.  {\em Unitary In-addresses} are denoted by lower case letters, with or without subscripts, and an inverted wedge prefix: $\cat\adrA_{1}, \dots, \cat\adrA_{m}$.  {\em Unitary Out-addresses} are denoted by lower case letters with a wedge prefix: $\pitcher\adrB_{1}, \dots, \pitcher\adrB_{k}$.  

An {\em address combination}  is a  sequence of zero or more unitary addresses.  
An {\em In-address combination} denoted by Greek letters with an inverted wedge prefix, e.g.,  $\cat\alpha =\, \cat^{m}\alpha =\, \cat\adrA_{1}, \dots, \cat\adrA_{m}$ is a sequence of zero or more unitary In-addresses.  The superscript $m$ denotes the length of the address combination.  
An {\em Out-address combination}, denoted by Greek letters with a wedge prefix, is a  sequence of zero or more unitary Out-addresses: $\pitcher\beta =\, \pitcher^{k}\beta =\, \pitcher\adrB_{1}, \dots, \pitcher\adrB_{k}$.  The superscript $k$ denotes the length of the address sequence.   

Each node $\nodeA$ in an addressing network $N$ has at least one In-address combination and one Out-address combination. The general form a node $\nodeA =\, \fiadrA_{1}\ldots\fiadrA_{i},X,\foadrB_{1},\dots, \foadrB_{o}$.  The {\em full  In-address} ({\em full  Out-address}) of a node is the sequence of unitary In-addresses (unitary Out-addresses) gotten by stringing together the In-address (Out-address) combinations of a node.  While the distinction between address combinations and their full  counterparts is useful for describing the general topology of addressing networks, in this article we use the full  description for the In- and Out-addresses of nodes. 

Let $\cat\nodeA$ denote the full {\em  In-address} of node $\nodeA$ which consists of a sequence of $m$ unitary addresses $\cat\nodeA =\, \cat^{m}\nodeA =\, (\cat\adrA_{1} \ldots \cat\adrA_{m})$ where the $\cat\adrA_{i}$ in $\cat\nodeA$ are the unitary In-addresses.  Let $\pitcher(\nodeA) =\, \pitcher\nodeA =\, \pitcher\,^{k}\nodeA$ denote a full {\em Out-address} of node $\nodeA$ which consists of a sequence of $k$ unitary Out-addresses:  $\pitcher\nodeA =\, (\pitcher\adrB_{1} \ldots \pitcher\adrB_{k})$.  The number of unitary addresses in a full address called the {\em address length} or {\em address size}.  A node may have more than one full Out-address.  

A unitary Out-address $\pitcher\adrA$ {\em matches}  a unitary In-address $\cat\adrB$ if $(\pitcher\adrA, \cat\adrB) \in \match$, i.e., if   $\match(\pitcher\adrA, \cat\adrB)$ holds.  
The matching relation $\match$ is specified externally by an {\em interpretive-executive system} (IES) that interprets and executes the network $N$.   Thus, addressing networks are executable networks that are interpreted and executed by some external system we call the IES.  Examples of addressing networks include mobile and non-mobile telephone systems, the Internet, the postal delivery service, and, as we shall see, gene regulatory networks (GRNs).  

Let $X \in \mathbb{X}$ denote some action directive.   A node $\nodeA$ in a addressing network $N$ has the general form:  $(\cat\nodeA, X, \pitcher\nodeA )$.  Ignoring the action $X$ component, a node $\nodeA$ with $m$ unitary In-addresses $\cat^{m}\nodeA $ and $k$ unitary Out-addresses $\pitcher^{k}\nodeA$ is denoted variously as  $(\cat\adrA_{1}\ldots\cat\adrA_{m}|\pitcher\adrB_{1}\ldots\pitcher\adrB_{k}) =\, \cat\,^{m} \nodeA\pitcher\,^{k} =\, ^{m}\nodeA^{k}$.  Note, since In- and Out-addresses are sequences and not sets of unitary addresses, two nodes with the same action and same In- and Out-addresses need not be identical. 
 
An {\em Out-node} in a network is any node with at least one unitary Out-address. An {\em In-node} is any node with at least one unitary In-address.  

A {\em unitary directed link} $(\cat\adrA\, \rightarrow\, \pitcher\adrB)$ is formed from node $\nodeA$ to node $\nodeB$ in network $N$ if $\exists\,\pitcher\adrA_{i} \in\, \pitcher\nodeA =\, (\pitcher\adrA_{1} \ldots \pitcher\adrA_{k})$ and $\exists\, \cat\adrB_{j}  \in\, \cat\nodeB =\,  (\cat\adrB_{1} \ldots \cat\adrB_{m}) $  such that $\match(\pitcher\adrA, \cat\adrB)$, i.e., $(\pitcher\adrA_{i}, \cat\adrB_{j}) \in \match$.  

All references to  links or edges will denote unitary directed links.  Note, a unitary In-address may match more than one unitary Out-address.  And, a unitary Out-address may match many unitary In-addresses.  Hence, multiple links may form between Out-nodes and In-nodes. 

Nodes with no unitary Out-addresses are called {\em terminal nodes} and denoted by $\catm\nodeA\pitcherZero =\, ^{m}\nodeA^{0} =\, !\nodeA$. Nodes with no unitary In-addresses are called {\em inaccessible nodes} denoted by $\catZero\nodeA\pitcherk =\, ^{0}\nodeA^{k}$.  For example, the node $^{0}\nodeA^{0}$ is both inaccessible and terminal.  The simplest accessible Out-node is of the form $\cat\,^{1} \nodeA\pitcher\,^{1} =\,^{1}\nodeA^{1} =\,  (\cat\adrA\, | \,\pitcher\adrB)$
where an Out-node $\nodeA$ has only one unitary In-address $\cat\adrA$ and only one unitary Out-address $\pitcher\adrB$.   
Given a node $\nodeA$ in a network $N$ with a nonempty In-address $\catm\nodeA$ for $m > 0$, if there exists no node in $N$ with a matching Out-address, then the node is {\em inaccessible within} $N$.  Such nodes may be accessible to external networks or signals. 

\section{Ordered and unordered address combinatorics} 
Addresses in an addressing network are formed by combinations of unitary Out-addresses and combinations of basic In-addresses.  Generally, in combinatorics given $n$ units that form combinations, if the units are ordered, e.g., where $(\cat\adrA,\cat\adrB,\cat\adrC) \neq (\cat\adrB,\cat\adrA,\cat\adrC)$ then for address combinations of length $k$ there are $n^k$ possible combinations.  If the address combination are unordered, e.g., where $(\cat\adrA,\cat\adrB,\cat\adrC) = (\cat\adrB,\cat\adrA,\cat\adrC)$, then there are  
$\left( \begin{array}{c} n \\ k \end{array} \right) =  \frac{n!}{k!(n - k)!}$ possible unordered address combinations.  Since number of possible links in an encoded addressing network $N$ is bounded by the number of possible addresses,  the large numbers of both ordered and unordered address combinations appear to be sufficient to enable the generation of large, complex networks. However, we will show that in the case of bounded encoded addressing networks these seemingly ample address combinations are illusory based on mistaken implicit, combinatorial presuppositions.  

\section{Combinatoric limits of encoded addressing networks} 

There are fundamental combinatorial properties that can limit the control capacity of encoded networks. 

Let $N^{E}$ be a sequential encoding of a network $N$ in a language $L$.  
If $\nodeA$ is a node in $N$ then $\nodeA^{E}$ is its encoding in $N^{E}$.  The encoded address $\match$ relationships determine the encoded links between nodes.  

Assume there are a finite number $n$ of unitary Out-addresses, $\pitcher\adrB_{1}\ldots\pitcher\adrB_{n}$, encoded in the network $N^{E}$.  Assume that each unitary Out-address $\pitcher\adrB$ contained in the set of unitary Out-addresses $\mathbb{O}$ of $N$
is encoded only once in $N^{E}$.     Assume each encoded Out-node $\nodeA^{E}$  
in  $N^{E}$ has an encoded Out-address $\pitcher\nodeA =\, \pitcherk\nodeA$ consisting of a combination of at least $k \geq 1$ unitary Out-addresses.  We now show that given theses assumptions  there are at most $n/k$ encoded Out-address combinations of length $k$ in $N^{E}$. Hence, by definition, there are at most $n/k$ encoded Out-nodes $\nodeA$ in the encoded network $N^{E}$.

\begin{thm} If $N^{E}$ contains $n$ encoded Out-addresses and if there are no repeats of encoded unitary Out-addresses in $N^{E}$ and if each Out-node contains at least $k$ unitary Out-addresses, then the maximum number of Out-address combinations in an encoded network $N^{E}$ is $n/k$. 
\label{thm:OutAdr}\end{thm}

\begin{proof}
Standard combinatorics assumes that the basic elements that form combinations can be repeated in combinations. Thus, normally it can be assumed that unitary addresses which are the elements that form address combinations can be repeated in those combinations. For example, $(\pitcher\adrA,\pitcher\adrB,\pitcher\adrC)$ and $(\pitcher\adrA,\pitcher\adrB,\pitcher\adrD)$ are different combinations. However, these combinations repeat both the unit $\pitcher\adrA$ and the unit $\pitcher\adrB$.  Under our assumption of no repeats of unitary Out-addresses, if $\pitcher\adrA$ and $\pitcher\adrB$ are encoded only once in an encoded network $N^{E}$ then there can be no encoding in $N^{E}$ of both combinations $(\pitcher\adrA,\pitcher\adrB,\pitcher\adrC)$ and $(\pitcher\adrA,\pitcher\adrB,\pitcher\adrD)$.  Hence, if $k$ is the minimum Out-address length of each node $\nodeA^{E}$ in $N^{E}$ and if $n$ is the total number of encoded unitary Out-addresses in $N^{E}$ then the encoded network $N^{E}$ contains at most $n/k$ encoded Out-address combinations. 
\end{proof}

\begin{cor} If $N^{E}$ contains $n$ encoded Out-addresses and if there are no repeats of encoded unitary Out-addresses in $N^{E}$ and if each Out-node contain at least $k$ unitary Out-addresses, then the maximum number of Out-nodes in an encoded network $N^{E}$ is $n/k$. 
\end{cor}

\begin{proof}
Follows immediately from Theorem \ref{thm:OutAdr} by definition of Out-node.
\end{proof}

If $k = 1$ there can be at most $n$ encoded Out-addresses, and $n$ Out-nodes each with only a single unitary Out-address.  
The Out-nodes of $N^{E}$ are called the {\em control nodes} of the network $N$ because only Out-nodes can initiate and direct action. They form the fundamental {\em control backbone} of the network.  Thus, given the assumptions above, the number of possible effective control nodes in $N^{E}$ is $n/k$.  The  {\em control capacity} of an encoded network $N^{E}$ is a function of the number of control nodes, i.e., Out-nodes in the network.  While the number of Out-nodes puts no limits on the number of In-nodes, it puts severe restrictions on the possible control capacity of the network $N^{E}$.   Note, these results hold for any encoded addressing network, not just for gene regulatory networks (GRNs) discussed below. 

\section{Combinatoric limits of virtual addressing networks}
Relative to a set of encoded In-addresses, the Out-nodes and links defined by the encoded Out-addresses form the encoded portion of the network $N$ which we call the {\em primary encoded network} $N^{E}$.  The question is to what extent can network addresses and links be formed during the execution of the network.  {\em Virtual addresses} and {\em virtual links} are addresses and links that are not explicitly encoded in $N^{E}$ and are instead generated as the network is executed by the IES.  We now show that a virtual network generated by combinations of encoded addresses cannot extend the control capacity of the primary encoded network. 

A {\em virtual address} $\pitcher\nodeV$ in a network is combination of unitary addresses not explicitly encoded as a sequence in some Out-node in the network.  The {\em virtual network} generated by a network $N^{E}$ consists of those links $(\pitcher\nodeV \rightarrow\, \cat\nodeD^{E})$ where the Out-address combination $\pitcher\nodeV$ is virtual and it matches the In-address $\cat\nodeD^{E}$ of some encoded In-node $\nodeD^{E}$ in $N^{E}$.  Let $\pitcher\nodeV = (\pitcher \adrA_{1}, \ldots, \pitcher \adrA_{k})$ be any virtual Out-address that is not encoded directly in $N^{E}$. 

\subsection{Informal Proof}
By assumption each unitary Out-address   $\pitcher\adrA_{i}$  in  $\pitcher\nodeV$ occurs once and only once in some encoded node $\nodeB$ in $N^{E}$.  
To generate the virtual Out-address combination $\pitcher\nodeV$  each unitary Out-address  $\pitcher\adrA_{i}$  in  $\pitcher\nodeV$ must be called by some Out-node $\nodeA$.  Consider $\pitcher\adrA_{i}$.  To generate $\pitcher\adrA_{i}$ either it occurs directly, encoded in $\nodeA$ (where $\pitcher\adrA_{i}$ is in the Out-address combination $\pitcher\nodeA$) or $\pitcher\adrA_{i}$ occurs in some other node and has to be called by an address $\pitcher \adrD$ contained in $\nodeA$'s Out-address $\pitcher\nodeA$. If $\pitcher\adrA_{i}$ is encoded in $\nodeA$ it cannot occur anywhere else in $N$. If $\pitcher\adrA_{i}$ is not encoded in $\nodeA$ it has to be called by some Out-address $\pitcher \adrD$ that is encoded in $\nodeA$ and the Out-address $\pitcher \adrD$ matches an In-address encoded in $\cat(\pitcher \adrA)$. Assume the match is sufficient to activate $\pitcher\adrA_{i}$, e.g., using OR-addressing. Similarly, for any other unitary Out-address $\pitcher \adrA_{j}$ in $\pitcher\nodeV$, either $\pitcher\adrA_{j} \neq \pitcher\adrA_{i}$ is encoded in $\nodeA$ or it has to be called by $\nodeA$. If called and  $\pitcher\adrA_{j}$ has the same In-address for $\pitcher \adrD$ as $\pitcher\adrA_{i}$ where $\pitcher d$ matches both $\cat(\pitcher\adrA_{i})$ and  $\cat(\pitcher\adrA_{j})$ then the generation of $\pitcher \adrD$ will generate the both unitary Out-addresses $\pitcher\adrA_{i}, \pitcher\adrA_{j}$.  If $\pitcher\adrA_{j}$ has a different In-addresses from $\pitcher\adrA_{i}$, then some Out-address $\pitcher\adrE$ that matches $\cat(\pitcher\adrA_{j})$ has to be encoded in $\nodeA$ or generated by $\nodeA$. Hence, for each combination address $(\pitcher\adrA_{1}, \ldots, \pitcher\adrA_{k})$ generated by $\pitcher(\nodeA)$ in $N$ if a unitary sub-address $\pitcher\adrX$  within the combination address $\pitcher\nodeV = (\pitcher \adrA_{1}, \ldots, \pitcher \adrA_{k})$ is not encoded in $\pitcher\nodeA$, it has to be generated by $\pitcher\nodeA$ with call to the node that generates $\pitcher\adrX$. If the activating In-address $\cat(\pitcher\adrX)$ of $\pitcher\adrX$ is different from the other unitary Out-addresses $\pitcher\adrA_{i}$ in $\pitcher\nodeV$ then such a call requires at least one more Out-address $\pitcher\adrY$ that matches an In-address in $\cat(\pitcher \adrX)$ to activate and generate $\pitcher\adrX$.  

\subsection{Formal Proof}

\begin{thm} If $N^{E}$ contains no loops and no signaling,  if $N^{E}$ contains $n$ encoded unitary Out-addresses and if there are no repeats of encoded unitary Out-addresses in $N^{E}$ 
then if a virtual address $\pitcher\nodeV$ of length $k \geq 2$ is generated dynamically during the execution  of the network, then maximum number of virtual address combinations that can be generated by an encoded network $N^{E}$ is $n/k$ and $k \geq 2$. 
\label{thm:VirtualNet}\end{thm}

\begin{proof}

Let $\pitcher\nodeV =\,  \pitcher^{k}\nodeV  = (\pitcher\adrV_{1}, \ldots, \pitcher\adrV_{k})$ be any virtual Out-address that is not encoded directly in $N^{E}$.  By definition of virtual node, there is no encoded node $\nodeA$ in $N^{E}$ that contains all the unitary Out-addresses in $\pitcher\nodeV$.  Hence, it requires at least  $2$ and up to $k$ encoded Out-nodes $\nodeA_{1} \ldots \nodeA_{k}$ to generate a virtual combination $\pitcher^{k}\nodeV$  such that each encoded node contains a subset of the unitary Out-addresses in the virtual address combination. Assume, without loss of generality, that two Out-nodes, $\nodeA$ and $\nodeB$, generate $\pitcher^{k}\nodeV$. 

Given a virtual Out-address $\pitcher^{k}\nodeV$ is generated by two Out-nodes $\nodeA$ and $\nodeB$, let $\pitcher^{x}\nodeV_{A}$ be the sub-address sequence generated by $\nodeA$ and $\pitcher^{y}\nodeV_{B}$ be the sub-address sequence generated by $\nodeB$. Since by assumption unitary Out-addresses are only encoded once, $\pitcher\nodeA$ cannot intersect $\pitcher\nodeB$.  Hence,  $x + y \geq k$ and $\nodeA$ and $\nodeB$ together generate the full virtual address $\pitcher^{k}\nodeV =\, \pitcher^{k}\nodeV_{A,B}$.  Thus, the generation of a virtual address $\pitcher^{k}\nodeV$ of size $k$ uses up $k$ unitary addresses.  By assumption, there are at most $n$ unitary addresses available in the network $N$. By definition, virtual Out-address consists of at least two unitary Out-addresses. Therefore, there are at most $n/2$ virtual addresses  can be generated by any (simple -no loops, no signaling) encoded network $N$.  More generally, if each virtual address is of size $\geq k$, then at most $n/k$ virtual combinations can generated by a network of size $n$ and $k \geq 2$. 
\end{proof}

\section{Gene Regulatory Networks as addressing networks}

Gene Regulatory Networks (GRNs) consist of transcription factor genes (TF-genes) that generate transcription factor proteins (TF-proteins) that bind to cis promoters (TF-promoters) of genes resulting in their possible activation. 
If  TF-genes are mapped  to unitary Out-addresses and TF-promoters are mapped to unitary In-addresses, then Gene Regulatory Networks (GRNs) can viewed as instances of addressing networks.  Gene Regulatory Networks are encoded linearly in genomes. Thus, GRNs are instances of linearly encoded addressing networks. 

The encoded links between nodes in GRNs consist of {\em transcription factor genes} (TF-genes)  and their cis promoter sequences (TF-promoters) that bind and catch matching TF-proteins generated by TF-genes.  TF-promoters are normally associated with one or more genes which are activated once their cis TF-promoters is loaded.  Thus, GRNs are addressing networks where the nodes of the network are linked by addresses that match in some way.  Combinations of TFs form the addresses of GRNs. TF-promoters, denoted by $\tfc_{i}$, combine to form  the In-addresses of nodes in GRNs. TF-genes are denoted by $\tfg$. Individual TF-genes, denoted by $\tfg_{j}$, correspond to the unitary Out-addresses of addressing networks.   TF-promoters are denoted by $\tfc$.   A node $\nodeA$ in a GRN has the general form  $\nodeA = (\tfc_{1} \ldots \tfc_{m}, X, \tfg_{1} \ldots \tfg_{k})$ with $m \geq 0$ and $k \geq 0$.   $\cat\nodeA = (\tfc_{1} \ldots \tfc_{m})$  is the node's In-address or cis promoter site and consists of zero or more TF-promoters $\tfc_{i}$. 
The Out-address  $\pitcher\nodeA = (\tfg_{1} \ldots \tfg_{m})$ consists of zero or more the TF-genes $\tfg_{j}$. 
$X$ is a, possibly null, cell action-directive.  The simplest linking node in a GRN has the form $\nodeA = (\tfc \,|\, \tfg)$ with an In-address $\cat\nodeA = (\tfc)$  consisting of a single TF-promoter $\tfc$ and a unitary Out-address $\pitcher\nodeA = (\tfg)$ consisting of single TF-gene $\tfg$.  

\subsection{Cis Promoter Logic}
We use the term TF-promoter for both cis regulatory promoters, repressors and activators (see \cite{Latchman2005,Hughes2011}). The activation of a particular node  $\nodeA$ with promoter $\catk\nodeA$ will depend on its cis-regulatory logic \cite{Davidson2006, Carroll2005, Carroll2008, Carroll2011, Furlong2012}. If it has AND-logic then all sub-addresses $\tfc_{i} \in\, \catk\nodeA$ must be loaded by their matching TF-protein $\tfp_{i}$.  If it has OR-logic then only one of the sub-addresses $\tfc_{i}$ needs to be loaded to activate the gene.  The cis-regulatory logic can be quite complex such as a Boolean function, or a threshold logic function.  
Nor does it matter that there appears to be no canonical address relationship between TFs and their cis promoters.  The nature of the cis-regulatory activation logic is independent of the combinatorial proof since it does not depend on the activation logic nor on the execution of the network by the IES.  All that is needed for the proof is that TF-genes and TF-promoters are encoded in the genome and form links by some matching relationship.  

\subsection{Consequences: Size limits of GRN networks} 
Given there are at most 1,000 TF-genes in extant genomes, then if the In-addresses of gene promoters would require just $k = 1$ matching TFs, then there are at most 1,000  control nodes in a pure GRN.  Hence, there would be at most 1,000 links in the network.  For a binary decision tree would have a depth of at most $9$.  $2^{9} = 512$ has $2^{n+1} - 1 = 1023$ Out-nodes and  $2^{n+1} - 2 = 1022$ Out-addresses or links.  For a network that controls the movement, division and differentiation of billions of cells, a network with only 1,000 control nodes and a depth of between 1,000 for a linear control path and 9 for a binary tree control structure, cannot generate the complex output sequences necessary for space-time control of the embryonic development of complex multicellular organisms.  Hence, the traditional theories of development and evolution based on GRNs cannot be adequate.  They cannot explain the control of such complex dynamic processes and they cannot explain the evolution of complex multicellular organism. 

\section{Control capacity of networks}
Let $N_{\nodeA}^{*}$ be the set of possible paths through a network starting from a node $\nodeA$.   If viewed in terms of action sequences that the paths in $N_{\nodeA}^{*}$ generate then  $N_{\nodeA}^{*}$ is the extensional representation of the {\em action strategy} $\pi(N_{\nodeA})$ of the network where $\pi^{*}(N_{\nodeA}) = N_{\nodeA}^{*}$.  
The {\em control capacity} of a network $N$ relative to a start node $\nodeA$ is a function of the number, length and complexity of possible paths in $N_{\nodeA}^{*}$.  The {\em generative capacity} of a network $N$ relative to a start node $\nodeA$ is a function of the maximally complex output that a path in $N_{\nodeA}^{*}$ can generate.

\subsection{Limits of cis evolutionary capacity}

Adding cis-promoters does not increase network size or control capacity.
The current network based view of how organisms evolve is that the cis promoters of genes evolve, while transcription factor genes are evolutionarily conserved over hundreds of millions of years \cite{Carroll2008,Carroll2005,Davidson2002,Davidson2006,Furlong2012}.  In the language of addressing networks, transformations of gene regulatory networks are limited to changes in the In-addresses of nodes. Thus, pure cis promoter evolution is restricted to In-address evolution and, therefore, cannot increase network size and capacity\footnote{Critique: Unless there exist Out-nodes with no matching In-nodes. Then adding In-addresses to inaccessible In-nodes can change the network topology and extend its connected functional size.}

This limits evolution to changes in topology of the network without increasing its size or capacity. The topology of a network $N$ can be transformed when In-addresses are modified. In-address transformations can result in novel developmental phenotypes. 

\subsection{Evolutionary capacity defined}

A {\em 1st order address operator} $\alpha^{1}$  on an addressing network $N$ changes a unitary address of some node in $N$ without changing the number of nodes in $N$.  A {\em 2nd order node operator} $\alpha^{2}$ on an addressing network $N$ adds to or deletes nodes from $N$.  A {\em 2nd order Out-address operator} on $N$ adds to or deletes Out-nodes from $N$.  1st order address operators result in transformations of network topology leaving the number of nodes constant. Combinations of 1st order address and generative 2nd order (copy/delete/replace) operators result in network transformations of topology, growth, complexity and capacity\footnote{Question: How does evolutionary capacity relate to control capacity?}.  
 
Let the {\em cis evolutionary capacity} $\cisM(N)$  be the set of all possible networks that can be generated from a given network $N$ if only the In-addresses of nodes in $N$ are changed, i.e., if only 1st order In-address transformations are allowed while Out-addresses are unchanged and the number $n$ of Out-nodes remains constant.   

\subsection{Invariance of control capacity under cis-transforms}
Note, all networks $N_{i} \in \cisM(N)$ have the same set of unitary Out-addresses and Out-nodes.  If some Out-nodes in $N$ are inaccessible in $N$ they may become accessible in some transform $N^{T} \in \cisM(N)$ leading to a greater control capacity.  However, if all Out-nodes are accessible in $N$ then the control backbone of any cis transformed network $N^{T} \in \cisM(N)$ remains invariant.  Hence, the maximal control capacity of the network under cis transforms remains invariant. 

No cis-network (In-address network) resulting from In-address operators on $N$, however complicated, can increase the combinatorial address capacity on an encoded network $N^{E}$.  While there is no restriction on repeating In-addresses, the restriction on Out-address combinations limits the control capacity of the network.   Regardless of the number of cis promoter In-addresses one adds to the network, it does not increase the Out-node number of the network.  All transformations, additions, or deletions of In-addresses can do is change to links and thereby the topology of the network and change the sets of terminal nodes that are linked in.  While this can significantly change the behavior of the network, it does not change the control backbone.  Thus, its ability grow in complexity is limited by constant size of the control backbone.  It cannot reflect the complexity of control needed to generate the complexity of space-time events that occur in embryogenesis and evolution.  It cannot grow in complexity in response to evolutionary pressures. It fundamentally limits the evolutionary capacity of the organism. 

\subsection{Non-additive 1st order trans evolutionary capacity}
An {\em 1st order Out-address operator (mutation)} of a network $N$ changes the Out-address $\pitcher\nodeA$ of Out-nodes $\nodeA$ in $N$ where Out-address transforms of $\pitcher\nodeA$  include modification of a given unitary Out-address, unitary Out-address additions and deletions .  A  1st order Out-address operator 
is an Out-address transformation that is non-additive and leaves the number of Out-nodes unchanged. It does not add Out-nodes by adding Out-addresses to terminal nodes. 

Let the {\em 1st order Out-address Evolutionary Capacity} $\transM(N)$ be the set of all possible networks that can be generated from an addressing network $N$ if only the Out-addresses of Out-nodes in $N$ are changed, i.e., if only 1st order Out-address transformations of Out-nodes are allowed.  By definition, 1st order Out-address transforms are non-additive leaving the number of Out-nodes invariant because they leave the terminal nodes with empty Out-addresses unchanged.  

Any Out-address transform that stays within the address space of a network $N$, except for addition or subtraction, can simulated by a sequence of In-address transforms of $N$. 

Question: Are the (1st order, 2nd order) In-address network manifolds and Out-address network manifolds equivalent?

\subsection{Additive 2nd order trans evolutionary capacity}
A {\em 2nd oder Out-address transformation} of a network $N$ modifies the Out-address any node $\nodeA \in N$,  including terminal nodes with empty Out-addresses, changing, adding to or deleting unitary Out-address from $\pitcher\nodeA$.  

Let the {\em 2nd order trans evolutionary capacity} or {\em Generative Evolutionary Capacity} $\metaM(N)$ of a network $N$ be the set of all possible networks that can be generated if Out-nodes can be created and added to the network $N$ such that the network's control backbone can grow and additive 2nd order Out-address transformations are allowed. 

The developmental capacity of a network both enables and limits the possible complexity its output.  The developmental capacity is bounded by the its control capacity which is defined by the number of Out-nodes in the network. The evolutionary capacity of a network depends on what kinds of network mutations or transformations are allowed.  Pure 1st order cis (In-address) and 1st order trans (Out-address) transformations  place inherent limits on the evolution of developmental network capacity and corresponding output complexity because they do not increase the number of Out-nodes in the network.  The evolution of complex organisms only becomes possible with 2nd order additive trans (Out-address) transformations that create and link new Out-nodes into the network.  Addition of Out-nodes enables the evolution of increase in network size and complexity which, in turn, allows a corresponding increase in the developmental capacity of evolving addressing networks. 

\section{Conclusion}

Given no loops or cycles and no random generation of Out-addresses, if all unitary Out-addresses in a virtual combination $\pitcher\nodeV$ have have the same In-address by which they can be activated by the same unitary Out-address then an encoded network $N^{E}$ with $n$ encoded Out-nodes, can generate at most $n$ different virtual address combinations.  If any two unitary Out-addresses in $\pitcher\nodeV$ require activation by different unitary Out-addresses, then if the minimum length of any virtual Out-address $\pitcher\nodeV$ is at least $k$ then an encoded network $N^{E}$ with $n$ encoded Out-nodes, can generate at most $n/k$ different virtual Out-address combinations.  

Therefore, each encoded unitary Out-address $\pitcher\adrX$ in a virtual address combination $\pitcher(\nodeV)$ generated by $\pitcher\nodeA$ (where the virtual address is encoded elsewhere and not in $\pitcher\nodeA$) has to be generated by means of a new encoded Out-address $\pitcher\adrY$.  Since, by assumption unitary Out-addresses, whether in combinations or not, are only encoded once in $N^{E}$, then since address combinations $(\pitcher\adrA_{1}, \ldots, \pitcher\adrA_{k})$ use unitary addresses repeatedly, most address combinations are {\em virtual} and not explicitly encoded in $N^{E}$.  Therefore, virtual address combinations have to be generated as the network is executed.  By the proof above, any virtual combination $\pitcher(\nodeV) = (\pitcher\adrA_{1}, \ldots, \pitcher\adrA_{k})$ (i.e., not encoded explicitly in $N$) requires at least one and up to $k$ new Out-addresses $(\pitcher\adrX_{1}, \ldots, \pitcher\adrX_{k})$ that match the In-addresses $(\cat(\pitcher\adrA_{1}), \ldots, \cat(\pitcher\adrA_{k}))$ of that combination respectively.  However, if there are only $n$ unitary Out-addresses in  $N^{E}$, there can be at most $n/k$ Out-address combinations of length $k$ available in $N^{E}$.   Hence, if each unitary address in a virtual address combination requires a distinct In-address then the network $N^{E}$ cannot have more than $n/k$ Out-address combinations be they explicit or virtual.  At best if  of all unitary addresses in a virtual combination have the same In-address, then there can be at most $n$ distinct virtual address combinations. 

Hence, the virtual network $N_{V}$ that consists of non-encoded Out-addresses that have combination addresses that repeat unitary Out-addresses, cannot be greater than the encoded network $N^{E}$.  In other words, the encoded network $N^{E}$ cannot generate a more complex, larger virtual address space needed for a larger virtual network.  
This means that transcription factor networks (GRNs) cannot by themselves create a large virtual address space. 

If the arguments are correct, then genes cannot explain development or the evolution of metazoans.

\footnotesize 
\bibliographystyle{abbrv}
\bibliography{WernerTFlimitsInformal}

\end{document}